\documentclass[reqno]{amsart}
\usepackage{amsmath,amsfonts,amsgen,amstext,amsbsy,amsopn,amsthm, amssymb}
\newtheorem{Lemma}{Lemma}
\newtheorem{Theorem}{THEOREM}

\newtheorem{corollary}{Corollary}
\theoremstyle{definition}

\theoremstyle{definition}

\theoremstyle{definition}

\renewcommand\epsilon\varepsilon

\newcommand{\T}{\mathbb{T}}
\newcommand{\Z}{\mathbb{Z}}

\newcommand{\be}{\begin{equation}}
\newcommand{\ee}{\end{equation}}
\newcommand{\bea}{\begin{align}}
\newcommand{\eea}{\end{align}}

\newcommand\hv{\widehat{v}}

\begin{document}
\title[The Excitation Spectrum for Weakly Interacting Bosons]{The Excitation Spectrum for \\ Weakly Interacting Bosons}
\author{Robert Seiringer}
\thanks{\copyright\ 2010 by the author. This work may be
reproduced, in its entirety, for non-commercial purposes.}
\address{Department of Mathematics and Statistics, McGill University, Burnside Hall, 805 Sherbrooke Street West, 
Montreal, Quebec H3A 2K6,
Canada}
\email{rseiring@math.mcgill.ca}
\date{Sept. 17, 2010}
\begin{abstract} 
  We investigate the low energy excitation spectrum of a Bose gas with
  weak, long range repulsive interactions. In particular, we prove that
  the Bogoliubov spectrum of elementary excitations with linear
  dispersion relation for small momentum becomes exact in the
  mean-field limit.
\end{abstract}

\maketitle

\section{Introduction and Main Results}

Bogoliubov's seminal 1947 paper \cite{bogol} contains several important
results concerning the low energy behavior of bosonic
systems. Among its striking predictions is the fact that the excitation
spectrum is made up of elementary excitations whose energy is linear
in the momentum for small momentum. Bogoliubov's method is based on
various approximations and crucially uses a formalism on Fock space
that does not conserve particle number. Mathematically, the validity
of his method has so far only been established for the ground state
energy of certain systems, see \cite{LSol,LSol2,Sol,ESY,GSlhy,yauyin}. In
particular, there are no rigorous results on the low energy excitation spectrum
of interacting Bose gases, with the notable exception of exactly solvable
models in one dimension \cite{gir,LL,L,calogero,sutherland}.

In this article, we shall prove the validity of Bogoliubov's
approximation scheme for a Bose gas in arbitrary dimension in the
mean-field (Hartree) limit, where the interaction strength is proportional to
the inverse particle number, and its range extends over the whole
system. In particular, we verify that the low energy excitation
spectrum for such a system equals the sum of elementary
excitations, as predicted by Bogoliubov. As a corollary, we observe
that the lowest energy in the sector of total momentum $P$ depends
linearly on $|P|$, a property that is crucial for the superfluid
behavior of the system. The mean-field limit has served as a
convenient and instructive toy model for several aspects of bosonic systems over the
years. We refer to \cite{FL,FKS} for a review and further references.

We consider a homogeneous system of $N \geq 2$ bosons on the flat unit
torus $\T^d$, $d\geq 1$. The bosons interact with a weak two-body interaction
which we write for convenience as $(N-1)^{-1} v(x)$. We assume that
$v$ is positive, bounded, periodic (with period one), and
$v(-x)=v(x)$. We also assume that $v$ is of positive type, i.e., it
has only non-negative Fourier coefficients. With $\Delta$ denoting the
usual Laplacian on $\T^d$, the Hamiltonian equals
$$
H_N = - \sum_{i=1}^N \Delta_i + \frac 1{N-1} \sum_{i<j}  v(x_i-x_j)
$$
in suitable units. 
It acts on $L^2_{\rm sym}(\T^{dN})$, the permutation-symmetric square integrable functions of $N$ variables $x_i\in \T^d$. 
Let $E_0(N)$ denote the ground state energy of $H_N$. The Bogoliubov approximation \cite{bogol,LSSY} predicts that $E_0(N)$ is close to $\tfrac 12 N \hv(0) + E^{\rm Bog}$, where
\begin{equation}\label{defebn}
E^{\rm Bog}=  -  \frac 12 \sum_{p\neq 0} \left( |p|^2 + \hv(p) - \sqrt{|p|^4 + 2|p|^2 \hv(p)} \right) \,.
\end{equation}
The sum runs over  $p \in (2\pi \Z)^d$, and 
$$
\hv(p) = \int_{\T^d} v(x) e^{-ipx} dx
$$
are the Fourier coefficients of $v$. 
Note that the sum above converges, since the summands behave like $\hv(p)^2/|p|^2$ for large $p$. 

More importantly, the Bogoliubov approximation predicts that the excitation spectrum of $H_N$ is made up of elementary excitations of momentum $p$ with corresponding energy
\begin{equation}\label{defep}
e_p = \sqrt{|p|^4 + 2 |p|^2 \hv(p)}\,.
\end{equation}
One noteworthy feature of (\ref{defep}) is that it is linear in $p$ for small $p$, in contrast to the case when interactions are absent.

Our main results can be summarized as follows. 

\begin{Theorem}\label{thm}
 The ground state energy $E_0(N)$ of $H_N$ equals
\begin{equation}\label{gse}
E_0(N) = \frac N 2 \hv(0) + E^{\rm Bog} + O(N^{-1/2})\,,
\end{equation}
with $E^{\rm Bog}$ defined in (\ref{defebn}). 
Moreover, the spectrum of $H_N- E_0(N)$ below an  energy $\xi$ is equal to finite sums of the form 
\begin{equation}\label{evf}
\sum_{p \in (2\pi\Z)^d \setminus\{0\}} e_p\, n_p  + O\left( \xi^{3/2} N^{-1/2}\right)\,, 
\end{equation}
where $e_p$ is given in (\ref{defep}) and $n_p \in \{0,1,2,\dots\}$
for all $p\neq 0$. 
\end{Theorem}

The error term $O(N^{-1/2})$ in (\ref{gse})  refers to an expression that is bounded,
in absolute value, by a constant times $N^{-1/2}$ for large $N$, where
the constant depends only on the interaction potential $v$; likewise
for the error term $O(\xi^{3/2}N^{-1/2})$ in (\ref{evf}). The
dependence on $v$ is rather complicated but can be deduced from our
proof, which gives explicit bounds. Keeping track of this dependence
allows to draw conclusions about the spectrum for large $N$ even if $v$ depends on
$N$.

In the case of a fixed, $N$-independent $v$, Theorem~\ref{thm} implies
that the Bogoliubov approximation becomes exact in the mean-field (Hartree) 
limit. As long as $\xi \ll N^{1/3}$, each individual excitation energy
$\xi$ is of the form $ \sum_p e_p\, n_p$ with error $o(1)$. Moreover,
as long as $\xi \ll N$, it is of the form $\sum e_p\,n_p (1+o(1))$,
i.e., the error is small relative to the magnitude of the excitation
energy. In other words, in the mean field limit the whole excitation
spectrum with energy $\xi\ll N$ is given in terms of sums of
elementary excitations. The condition $\xi\ll N$ can be expected to be
optimal, since only under this condition a large fraction of the
particles are guaranteed to occupy the zero momentum mode, one of the
key assumptions in the Bogoliubov approximation. For excitation
energies of the order $N$ and larger, the spectrum will not be
composed of elementary excitations anymore but has a more complicated
structure.

Our proof  shows that for each value of the $\{n_p\}$
there exists exactly one eigenvalue of the form (\ref{evf}). Moreover,
the eigenfunction corresponding to an eigenvalue with given $\{n_p\}$
has total momentum $\sum_{p} p\, n_p$. Given this fact we readily
deduce the following corollary from Theorem~\ref{thm}.

\begin{corollary}
Let $E_P(N)$ denote the ground state energy of $H_N$ in the sector of total momentum $P$. We have
\begin{equation}\label{epd}
E_P(N) - E_0(N) = \min_{\{n_p\},\, \sum_p p\, n_p = P} \sum_{p\neq 0} e_p\, n_p + O\left(|P|^{3/2}N^{-1/2}\right)\,. 
\end{equation}
In particular, $E_P(N) - E_0(N) \geq |P| \min_{p} \sqrt{2 \hv(p) + |p|^2} + O(|P|^{3/2}N^{-1/2})$.
\end{corollary}

The linear dependence of $E_P(N)$ on $|P|$ is of crucial importance
for the superfluid behavior of the Bose gas; see, e.g., the detailed
discussion in \cite{CDZ}.  It is a result of the interactions among
the particles. The expression (\ref{epd}) differs markedly from the
corresponding result for an ideal, non-interacting gas, especially if
$\hv(0)$ is large.

Finally, we note that under the unitary transformation $U=\exp(-i q
\sum_{j=1}^N x_j)$, $q\in (2\pi\Z)^d$, the Hamiltonian $H_N$
transforms as
$$
U^\dagger H_N U = H_N + N |q|^2 - 2 q  P\,,
$$
where $P = -i \sum_{j=1}^N \nabla_j$ denotes the total momentum
operator. Hence our results apply equally also to the parts 
of the spectrum of $H_N$ with excitation energies close to $N|q|^2$,
corresponding to collective excitations where the particles move
uniformly with momentum $q$.

The remainder of this paper is devoted to the proof of
Theorem~\ref{thm}. The main strategy is to compare $H_N$ with
Bogoliubov's approximate Hamiltonian. The latter has to be suitably
modified to take particle number conservation into account.

\section{Preliminaries}

We denote by $P$ the projection onto the constant function in $L^2(\T^d)$, and $Q=1-P$. The operator that counts the number of particles outside the zero momentum mode will be denoted by $N^>$, i.e.,
$$
N^> = \sum_{i=1}^N Q_i\,.
$$
We shall also use the symbol $T$ for the kinetic energy $ T = - \sum_{i=1}^N \Delta_i$. 

Lemma~\ref{lem1} below gives simple upper and lower bounds on the ground state energy $E_0(N)$ of $H_N$, as well as an upper bound on the expectation value of $T$ in a low energy state.

\begin{Lemma}\label{lem1}
The ground state energy of $H_N$ satisfies the bounds
\begin{equation}\label{sib}
0 \geq  E_0(N)  - \frac N 2 \hv(0) \geq -  \frac {N}{2(N-1)} \left( v(0) - \hv(0)\right) \,.
\end{equation}
Moreover, in any $N$-particle state $\Psi$ with $\langle \Psi|H_N|\Psi\rangle \leq \frac N2 \hv(0) + \mu$ we have 
\begin{equation}\label{aprn1}
(2\pi)^2 \left\langle \Psi | N^> | \Psi\right\rangle \leq \left\langle \Psi \left |  T  \right|\Psi\right\rangle \leq  \frac {N}{2(N-1)}\left( v(0) -\hv(0)\right) + \mu \,.
\end{equation}
\end{Lemma}

\begin{proof}
The upper bound to the ground state energy follows from using  a constant trial function. 
Since $\hv\geq 0$, we have
$$
\sum_{p\in (2\pi\Z)^d\setminus\{0\}} \hv(p) \left| \sum_{j=1}^N e^{ipx_j}\right|^2 \geq 0\,.
$$
This inequality can be rewritten as 
\begin{equation}\label{pdv}
 \sum_{1\leq i<j\leq N} v(x_i-x_j) \geq \frac{N^2}{2} \hv(0)  - \frac N 2 v(0) 
\end{equation}
and thus
$$
H_N \geq \frac N2 \hv(0) + T  - \frac{N}{2(N-1)}\left( v(0) -\hv(0)\right)\,.
$$
The rest follows easily.
\end{proof}

In the following, we shall also need a bound on the expectation value of higher powers of $N^>$. More precisely, we shall use the following lemma.

\begin{Lemma}\label{lem:apr}
  Let $\Psi$ be an $N$-particle wave function in the spectral subspace
  of $H_N$ corresponding to energy $E \leq E_0(N) + \mu$. Then
\begin{align*}
(2\pi)^2 \left\langle \Psi \left| N^> T \right|\Psi\right\rangle  & \leq  \left(v(0)+ \frac \mu 2 \right)^2 \\ & \quad + \frac{N}{2(N-1)} \left( \mu + 3 v(0) + \hv(0)\right)\left( 2\mu + v(0) -\hv(0)\right) \,.
\end{align*}
\end{Lemma}

In particular, $\langle\Psi| N^> T|\Psi\rangle$  is bounded above by an expression depending only on $\mu$ and $v$ but not on $N$.

\begin{proof}
Since $\Psi$ is permutation symmetric, 
$$
\left\langle \Psi \left| N^> T \right|\Psi\right\rangle =  N \left\langle \Psi \left|  Q_1   S \right| \Psi\right\rangle + \left\langle \Psi \left| N^> \left(H_N - E_0(N) - \tfrac 12 \mu\right)\right|\Psi\right\rangle
$$
where $S= E_0(N) + \tfrac 12 \mu  - (N-1)^{-1} \sum_{i<j} v(x_i-x_j)$. Using Schwarz's inequality, the last term can be bounded as 
$$
 \left\langle \Psi \left| N^> \left(H_N - E_0(N) - \tfrac 12 \mu\right)\right|\Psi\right\rangle \leq  \frac \mu 2 \left\langle \Psi \left| (N^>)^2 \right|\Psi\right\rangle^{1/2}\,.
$$

We split $S$ into two parts,  $S= S_a + S_b$, with
$$
S_a =  E_0(N) + \frac \mu 2  - \frac 1{N-1} \sum_{2\leq i<j\leq N} v(x_i-x_j)
$$
and 
$$
S_{b} =   - \frac 1{N-1} \sum_{j=2}^N v(x_1-x_j) \,.
$$
Note that $S_a$  does not depend on $x_1$. Using the positivity of $\hv(p)$ as in (\ref{pdv}), but with $N$ replaced by $N-1$, as well as the upper bound on $E_0(N)$ in (\ref{sib}), we see that
$$
S_{a} \leq \tfrac 12 \left( \mu +   \hv(0) +   v(0)\right) \,.
$$
In particular, this implies that
$$
 N \left\langle \Psi \left|  Q_1   S_a \right| \Psi\right\rangle  \leq \tfrac 12 \left( \mu +  \hv(0) +   v(0) \right)  \left\langle \Psi \left| N^> \right| \Psi\right\rangle\,.
$$
To bound the contribution of $S_b$, we use
\begin{align*}
- \langle \Psi| Q_1 S_b |\Psi\rangle =  \left\langle \Psi \left|  Q_1  v(x_1-x_2)  \right| \Psi\right\rangle & =  \left\langle \Psi \left|  Q_1 Q_2  v(x_1-x_2)  \right| \Psi\right\rangle \\ & \quad +  \left\langle \Psi \left|  Q_1  P_2  v(x_1-x_2) P_2 \right| \Psi\right\rangle \\ & \quad +\left\langle \Psi \left|  Q_1 P_2  v(x_1-x_2) Q_2 \right| \Psi\right\rangle \,.
\end{align*}
The second term on the right side is positive. For the first and the third, we use Schwarz's inequality and $\|v\|_\infty = v(0)$  to conclude
$$
\langle\Psi | Q_1 S_b |\Psi\rangle \leq v(0) \langle \Psi | Q_1 Q_2 |\Psi\rangle^{1/2}   +v(0) \langle \Psi | Q_1 | \Psi\rangle\,.
$$

Since 
$$
N^2 \langle \Psi| Q_1 Q_2 |\Psi\rangle \leq \langle \Psi |(N^>)^2 |\Psi\rangle
$$
we have thus shown that 
\begin{align*}
 \left\langle \Psi \left| N^> T \right|\Psi\right\rangle  & \leq \tfrac 12 \left( \mu +  \hv(0) + 3 v(0) \right) \left\langle \Psi \left| N^> \right|\Psi\right\rangle  \\ & \quad + \left( v(0) + \tfrac 12 \mu \right) \left\langle \Psi \left| (N^>)^2  \right|\Psi\right\rangle^{1/2} \,.
\end{align*}
Using $N^>\leq (2\pi)^{-2} T$ this yields
$$
 \left\langle \Psi \left| N^> T \right|\Psi\right\rangle \leq \left(\frac{v(0)+ \tfrac 12 \mu}{2\pi}\right)^2 + \left( \mu + 3 v(0) + \hv(0)\right)  \left\langle \Psi \left| N^> \right|\Psi\right\rangle\,.
$$
The result then follows from Lemma~\ref{lem1}.
\end{proof}

\section{The Bogoliubov Hamiltonian}

The main strategy in the proof of Theorem~\ref{thm} is to compare the Hamiltonian $H_N$ with the Bogoliubov Hamiltonian. We will use a slightly modified version of it, which is, in particular, particle number conserving. 

Let $a_p$ and $a_p^\dagger$ denote the usual creation and annihilation operators on Fock space for a particle with momentum $p\in (2\pi \Z)^d$, satisfying the canonical commutation relations $[a_p,a^\dagger_q]=\delta_{pq}$. 
For $p\neq 0$, let 
\begin{equation}\label{defb}
b_p = \frac{a_p a^\dagger_0}{\sqrt{N-1}} 
\end{equation}
and define the Bogoliubov Hamiltonian
\begin{equation}\label{bogham}
H^{\rm Bog} =  \sum_{p\neq 0} \left[ |p|^2 \,b^\dagger_p b_p  + \frac 12   \hv(p) \left( 2  b^\dagger_p b_p  + b^\dagger_p b^\dagger_{-p} + b_{p} b_{-p}  \right) \right]\,.
\end{equation}
Note that $H^{\rm Bog}$ conserves particle number. 
We are interested in $H^{\rm Bog}$ in the sector of exactly $N$ particles. Note also that $\hv(-p)=\hv(p)$ for all $p\in (2\pi \Z)^d$. 

Let 
$$
A_p =  |p|^2 + \hv(p)
$$
and 
$$
B_p =  \hv(p)\,.
$$
A simple computation (compare with \cite[Thm.~6.3]{LSol}) shows that 
\begin{align*}
& A_p \left( b^\dagger_p b_p + b^\dagger_{-p}b_{-p} \right) + B_p \left( b^\dagger_p b^\dagger_{-p} + b_p b_{-p}\right) \\ & =   \sqrt{A_p^2 - B_p^2} \left( \frac{ \left( b^\dagger_p + \alpha_p b_{-p}\right)\left( b_p + \alpha_p b^\dagger_{-p} \right) }{1-\alpha_p^2}+ \frac{ \left( b^\dagger_{-p} + \alpha_p b_{p}\right)\left( b_{-p} + \alpha_p b^\dagger_{p} \right) }{1-\alpha_p^2} \right)\\ & \quad - \frac 12 \left( A_p - \sqrt{A_p^2 - B_p^2} \right) \left( [b_p,b^\dagger_p] + [b_{-p},b^\dagger_{-p}]\right) \,,
\end{align*}
where
$$
\alpha_p = \frac {1}{B_p}\left( A_p -\sqrt{A_p^2 - B_p^2}\right) \quad \text{if $B_p>0$}\ , \quad \alpha_p=0 \quad \text{if $B_p=0$.}
$$ 
Note that $0\leq \alpha_p \leq \hv(p)/(|p|^2 + \hv(p))$. In particular,  $\sup_{p\neq 0} \alpha_p < 1$, and $\alpha(p) \sim \hv(p)/|p|^2$ for large $p$.

Define further 
$$
c_p = \frac{ b_p + \alpha_p b^\dagger_{-p}} {\sqrt{1-\alpha_p^2}} 
$$
for $p\neq 0$. 
What the above calculation shows is that 
\begin{equation}\label{hb1}
H^{\rm Bog} =  -  \frac 12 \sum_{p\neq 0} \left( A_p - \sqrt{A_p^2 - B_p^2} \right) \frac {[b_p,b^\dagger_p] + [b_{-p},b^\dagger_{-p}]}{2} +  \sum_{p\neq 0} e_p\, c^\dagger_p c_p
\end{equation}
with $e_p$ defined in (\ref{defep}). 
The commutators equal
\begin{equation}\label{hb2}
[b_p,b^\dagger_p] = \frac{a^\dagger_0 a_0  - a_p^\dagger a_p}{N-1} \leq \frac{N}{N-1}\,.
\end{equation}

\section{Proof of Theorem~\ref{thm}: Lower Bound}

Recall that $P$ denotes the projection onto the constant function in $L^2(\T^d)$, and $Q=1-P$. Denote the two-particle multiplication operator $v(x_1-x_2)$ by $v$ for short. Using translation invariance and the Schwarz inequality
\begin{multline*}
(P\otimes Q + Q\otimes P) v Q\otimes Q + Q \otimes Q v (P\otimes Q + Q \otimes P) \\ \geq - \epsilon (P\otimes Q + Q\otimes P)v(P\otimes Q + Q\otimes P) - \epsilon^{-1} Q\otimes Q v Q \otimes Q
\end{multline*}
(which follows from positivity of $v$) we conclude that 
\begin{multline}\label{s1}
v \geq P\otimes P v P\otimes P + P \otimes P v Q\otimes Q + Q\otimes Q v P\otimes P \\
+ (1-\epsilon) (P\otimes Q + Q\otimes P)v(P\otimes Q + Q\otimes P) - \epsilon^{-1} Q\otimes Q v Q \otimes Q 
\end{multline}
for any $\epsilon>0$. The last term can be bounded by $v(0) Q\otimes Q$. In second quantized language, this means that $H_N$ is bounded from below by the restriction of 
\begin{multline*}
\sum_p |p|^2 a^\dagger_p a_p + \frac {\hv(0)}{2 (N-1)} \left( N(N-1) - 2\epsilon (N-N^>)N^> -  N^>\left( N^>-1 \right) \right) \\ + \sum_{p\neq 0} \frac{ \hv(p)}{2} \left( 2 (1-\epsilon) b^\dagger_p b_p  + b^\dagger_p b^\dagger_{-p} + b_{p} b_{-p}  \right) - \frac {N^>(N^>-1) v(0)}{2 \epsilon (N-1)}
\end{multline*}
to the $N$-particle sector. Here, we use again the definition (\ref{defb}) of $b_p$. From now on, we shall work with operators on Fock space, but it is always understood that we are only concerned with the sector of $N$ particles.

Next, we observe that 
$$
a^\dagger_p a_p \geq \frac {N-1}{N} b^\dagger_p b_p\,.
$$
Moreover,
$$
\sum_{p\neq 0} \hv(p) b^\dagger_p b_p \leq \hv(0) \frac{N}{N-1} N^>\,,
$$
and hence 
$$
H_N \geq \frac N 2 \hv(0) + H^{\rm Bog} - E_\epsilon
$$
where the Bogoliubov Hamiltonian $H^{\rm Bog}$ was defined in (\ref{bogham}) and 
\begin{equation}\label{defee}
E_\epsilon = \frac 1{N-1} T  + \frac {N^>(N^>-1)}{2 (N-1)}\left( \hv(0)  + \frac{v(0)}{\epsilon}\right) + \epsilon \hv(0) \frac{2N-1}{N-1} N^> \,.
\end{equation}
In particular, using (\ref{hb1}) and (\ref{hb2}),
\begin{equation}\label{last}
H_N + E_\epsilon \geq  \frac N2 \hv(0) + \frac{N}{N-1} E^{\rm Bog} + \sum_{p\neq 0} e_p \, c^\dagger_p c_p\,,
\end{equation}
where $E^{\rm Bog}$ is the Bogoliubov energy defined in (\ref{defebn}).

The last term on the right side of (\ref{last}) is positive and can be
dropped for a lower bound on the ground state energy of $H_N$. For the
choice $\epsilon=O(N^{-1/2})$, the expected value of $E_\epsilon$ in
the ground state of $H_N$ is bounded above by $O(N^{-1/2})$, as the
bounds in Lemma~\ref{lem1} and~\ref{lem:apr} show. This proves the
desired lower bound on $E_0(N)$.

To obtain lower bounds on excited eigenvalues, it remains to investigate the positive last term in (\ref{last}). We do this via a unitary transformation.
Let $U=e^{X}$, where 
$$
X = \sum_{p\neq 0} \beta_p\left(b^\dagger_p b^\dagger_{-p}- b_{p}b_{-p}\right)
$$
with $\beta_p\geq 0$ determined by 
$$
\tanh(2 \beta_p) = \alpha_p\,.
$$
Note that $X$ is anti-hermitian and hence $U$ is unitary. If $a_0$ and
$a_0^\dagger$ were replaced by $\sqrt{N-1}$, $U$ would be the usual
Bogoliubov transformation. Our modified $U$ has the advantage of being
particle number conserving, however. The price to pay for this
modification is that $U^\dagger a_q U$ can not be calculated
anymore so easily.

A second order Taylor expansion yields
$$
e^{-tX} a_q e^{tX} = a_q - t [X,a_q] + \int_{0}^t (t-s) e^{-sX} [X,[X,a_q]] e^{sX} ds
$$
for any $t>0$.
We compute 
$$
[X,a_q]= - \frac 2 {N-1} \beta_q a_{-q}^\dagger  a_0^2
$$
and
\begin{align}\nonumber 
[X,[X,a_q]] &= 4 \beta_q^2 a_q \frac{a_0^2 (a_0^\dagger)^2}{(N-1)^2} - 4 \beta_q \left(\sum_{p\neq 0} \beta_p a_p a_{-p}\right) a_{-q}^\dagger \frac{2a^\dagger_0a_0+1}{(N-1)^2} \\ & =: 4 \beta_q^2 a_q + J_q \,. \label{defj}
\end{align}
For any $t>0$ we thus have
$$
e^{-t X} a_q e^{tX} = a_q +  \frac{2 t}{N-1}  \beta_q a_{-q}^\dagger  a_0^2 + \int_{0}^t (t-s) e^{-sX} \left( 4 \beta_q^2 a_q + J_q \right) e^{s X} ds\,.
$$
Iterating this identity leads to 
$$
U^\dagger a_q U  = \cosh(2\beta_q) a_q + \sinh(2\beta_q) a_{-q}^\dagger \frac{a_0^2}{N-1}+ K_q   =: d_q + K_q
$$
with 
$$
K_q =  \int_0^1 e^{-s X} J_q e^{sX} \frac{\sinh(2\beta_q (1-s))}{2\beta_q} ds\,.
$$

In particular, we see that 
\begin{align}\nonumber
U^\dagger a^\dagger_q a_qU & = d^\dagger_q d_q + K^\dagger_q K_q + d^\dagger_q K_q  + K^\dagger_q d_q
\\ & \leq (1+\lambda) d^\dagger_qd_q + (1+\lambda^{-1}) K^\dagger_q K_q \label{defl}
\end{align}
for any $\lambda>0$.  We further have  
\begin{align}\nonumber
 d_p^\dagger d_p & =  c^\dagger_p c_p   - \frac{a^\dagger_p a_p}{1-\alpha_p^2} \left( \frac{a_0a_0^\dagger}{N-1}-1 \right) - \frac{\alpha_p^2\, a_{-p}a^\dagger_{-p}}{1-\alpha_p^2} \left(  \frac{ a_0^\dagger a_0}{N-1}-\frac{ (a_0^\dagger)^2 a_0^2 }{(N-1)^2}\right) 
\\ & \leq c^\dagger_p c_p + \frac{a^\dagger_p a_p}{1-\alpha_p^2} \frac{N^>}{N} \,. \label{y2}
\end{align}
Using Schwarz's inequality, we can bound $K^\dagger_q K_q$ as  
\begin{equation}\label{comb1}
K^\dagger_q K_q \leq \frac{\cosh(2\beta_q) -1}{(2\beta_q)^2}\int_0^1 e^{-sX} J_q^\dagger J_q e^{sX} \frac{\sinh(2\beta_q (1-s))}{2\beta_q} ds\,.
\end{equation}
To get an upper bound on $J^\dagger_q J_q$, we write $J_q$ as the sum of two terms, $J_q^{(1)}+J_q^{(2)}$, where $J_q^{(2)}$ is the second term on the right side in the first line of (\ref{defj}), and $J^{(1)}_q = 4\beta_q^2 a_q((N-1)^{-2} a_0^2 (a_0^\dagger)^2-1)$. Using $J_q^\dagger J_q \leq 2 J_q^{(1)\dagger} J_q^{(1)} + 2 J_q^{(2)\dagger} J_q^{(2)}$ as well as
$$
\left(\sum_{p\neq 0} \beta_p a^\dagger_p a^\dagger_{-p}\right) \left(\sum_{q\neq 0} \beta_q a_q a_{-q}\right) \leq \left(\sum_{p\neq 0} \beta_p^2\right) N^>(N^>-1)\,,
$$
we obtain the bound 
\begin{equation}\label{comb2}
J^\dagger_q J_q  \leq  C_1 \beta_q^2 \frac{(N^>+1)^2}{N-1} \quad , \quad C_1 = 64 \frac N{N-1} \left( \sup_{q\neq 0} \beta_q^2 +  \sum_{p\neq 0} \beta_p^2\right)\,.   
\end{equation}

To proceed, we need an upper bound on $e^{-sX} (N^>+1)^2 e^{sX}$ for $0\leq s\leq 1$. For this purpose, let us  compute
$$
[X,N^> ] = -2 \sum_{q\neq 0} \beta_q \left( b^\dagger_q b^\dagger_{-q} + b_q b_{-q}\right)\,.
$$
We have
\begin{align}\nonumber
[X,N^>]^2 & \leq 8 \left(\sum_{p\neq 0} \beta_p b^\dagger_p b^\dagger_{-p} \right) \left(\sum_{q\neq 0} \beta_q b_q b_{-q}\right) + 8 \left(\sum_{p\neq 0} \beta_p b_p b_{-p} \right) \left(\sum_{q\neq 0} \beta_q b^\dagger_q b^\dagger_{-q}\right) \\  \nonumber & \leq \frac 8{(N-1)^2} \left(\sum_p \beta_p^2\right) \big[ N^>(N^>-1)(N+1-N^>)(N+2-N^>) \\ \nonumber & \qquad \qquad \qquad \qquad  \qquad + (N^>+1)(N^>+2)(N-N^>)(N-1-N^>)\big] \\ & \leq  16 \left(\sum_{p} \beta_p^2\right) \frac{N}{N-1} \left(N^>+1\right)^2\,. \label{xt}
\end{align}
With the aid of  Schwarz's inequality, we obtain
\begin{align*}
[X,(N^>+1)^2] &= (N^>+1) [ X, N^>] + [X,N^>](N^>+1) \\ & \geq -\eta (N^>+1)^2 - \eta^{-1} [X,N^>]^2
\end{align*}
for any $\eta>0$. In particular,
$$
[X,(N^>+1)^2] \geq - C_2 (N^>+1)^2
$$
with
$$
C_2 = 8 \sqrt{\frac N{N-1}\sum\nolimits_{q\neq 0} \beta_q^2}\,.
$$
We conclude that 
$$
e^{-tX} (N^>+1)^2 e^{tX} \leq (N^>+1)^2  + C_2 \int_0^t e^{-sX} (N^>+1)^2 e^{sX} ds
$$
for any $t>0$. 
Iterating this bound gives
\begin{equation}\label{n2b}
e^{-tX}(N^>+1)^2e^{tX} \leq e^{tC_2} (N^>+1)^2 \,.
\end{equation}
In combination, (\ref{comb1}), (\ref{comb2}) and (\ref{n2b}) yield the bound
\begin{equation}\label{y1}
K_q^\dagger K_q \leq  \frac{(N^>+1)^2}{N-1} C_1 e^{C_2}   \frac{(\cosh(2\beta_q)-1)^2}{16\beta_q^2}  \,.
\end{equation}

By combining (\ref{defl}) with (\ref{y2}) and (\ref{y1}), we obtain
\begin{align}\nonumber
\sum_p e_p\,  c_p^\dagger c_p & \geq  \frac 1{1+\lambda}  U^\dagger \left( \sum_p e_p\, a^\dagger_p a_p\right)  U- \frac{N^> T}{N}  \left( \sup_{p\neq 0} \frac{e_p}{|p|^2 (1-\alpha_p^2)}\right) \\ & \quad - (1+\lambda^{-1}) \frac{(N^>+1)^2}{N-1} C_1 e^{C_2} \left(\sum_q e_q  \frac{(\cosh(2\beta_q)-1)^2}{16\beta_q^2} \right)\,. \label{defll}
\end{align} 
Note that $|q|^2 \beta_q^2 \sim \hv(q)^2/|q|^2$ for large $q$, hence the last sum is finite.
Applying this bound to (\ref{last}), we have thus shown that 
$$
H_N + \widetilde E_{\epsilon,\lambda} \geq \frac N2 \hv(0) +  \frac{N}{N-1}E^{\rm Bog}+  \frac 1{1+\lambda} U^\dagger\left( \sum_{p} e_p\, a^\dagger_p a_p \right) U
$$
where
\begin{align*}
\widetilde E_{\epsilon,\lambda} & = E_\epsilon +  \frac{N^> T}{N}  \left( \sup_{p\neq 0} \frac{e_p}{|p|^2 (1-\alpha_p^2)}\right)  \\ & \quad + (1+\lambda^{-1}) \frac{(N^>+1)^2}{N-1} C_1 e^{C_2} \left(\sum_q e_q  \frac{(\cosh(2\beta_q)-1)^2}{16\beta_q^2} \right)\,,
\end{align*}
with $E_\epsilon$ defined in (\ref{defee}). 

The desired lower bound now follows easily from the min-max principle,
and the fact that the spectrum of $\sum_p e_p a^\dagger_p a_p$ equals
$\sum_p e_p n_p$, with $n_p \in \{0,1,2,\dots\}$ for all $p\in
(2\pi\Z)^d$. In fact, for any function $\Psi$ in the spectral subspace of $H_N$
corresponding to energy $E\leq E_0(N) + \xi$, we have
$$
\langle \Psi| \widetilde E_{\epsilon,\lambda} |\Psi\rangle \leq O\left( \left( \epsilon+N^{-1}\right)\xi +\xi^2 N^{-1} \left( \epsilon^{-1} + \lambda^{-1} \right)\right)
$$
according to Lemmas~\ref{lem1} and~\ref{lem:apr}. 
The choice $\epsilon = O( \sqrt{\xi/N}) =\lambda$ then leads to the conclusion that the spectrum $H_N$ below an energy $E_0(N)+\xi$ is bounded from  below  by 
the corresponding spectrum of 
$$
\frac N 2 \hv(0) + E^{\rm Bog} + \sum_{p\neq 0} e_p \, a^\dagger_p a_p  - O\left(\xi^{3/2}N^{-1/2}\right)\,.
$$
This completes the proof of the lower bound. 

\section{Proof of Theorem~\ref{thm}: Upper Bound}

We proceed in essentially the same way as in the lower bound. In analogy to (\ref{s1}), we have
\begin{multline*}
v \leq P\otimes P v P\otimes P + P \otimes P v Q\otimes Q + Q\otimes Q v P\otimes P \\ + (1+\epsilon) (P\otimes Q + Q\otimes P)v(P\otimes Q + Q\otimes P) + (1+\epsilon^{-1}) Q\otimes Q v Q \otimes Q
\end{multline*}
for any $\epsilon>0$. Together with
$$
b^\dagger_p b_p \geq a^\dagger_p a_p \left( 1- \frac {N^>}N\right)
$$
this implies the upper bound
$$
H_N \leq \frac N2 \hv(0) + H^{\rm Bog} + F_\epsilon
$$
where
\begin{equation}\label{fe}
F_\epsilon = \frac {N^>}{N} T  + \epsilon \hv(0) \frac{2N-1}{N-1} N^> + (1+\epsilon^{-1}) \frac {N^>(N^>-1) v(0)}{2 (N-1)} \,.
\end{equation}

From a lower bound to the commutator (\ref{hb2}), namely
$$
[b_p,b^\dagger_p] + [b_{-p},b^\dagger_{-p}] = \frac{2 a^\dagger_0 a_0  -  a^\dagger_p a_p -  a^\dagger_{-p}a_{-p}}{N-1} \geq 2 - \frac{3 N^>}{N-1} \,,
$$
we get 
\begin{equation}\label{u1}
H^{\rm Bog} \leq   E^{\rm Bog}\left( 1- \frac{3 N^>}{2(N-1)}\right)  + \sum_{p\neq 0} e_p\, c^\dagger_p c_p \,.
\end{equation}
Finally, to investigate the last term on the right side of this expression, we proceed as in (\ref{defl})--(\ref{defll}), with the obvious modifications to get an upper bound instead of a lower bound. In replacement of (\ref{y2}) we use 
$$
 d_p^\dagger d_p  \geq  c^\dagger_p c_p   - \frac 1{N-1} \frac{a^\dagger_p a_p}{1-\alpha_p^2} - \frac{\alpha_p^2}{1-\alpha_p^2} \frac {N^>(N^>+1)}{N-1}\,.
$$
The result is
\begin{align}\nonumber
\sum_p e_p\,  c_p^\dagger c_p & \leq  \frac 1{1-\lambda}  U^\dagger \left( \sum_p e_p\, a^\dagger_p a_p\right)  U + \frac{N^>(N^>+1)}{N-1}  \left( \sum_{p\neq 0} \frac{e_p \alpha_p^2}{1-\alpha_p^2}\right)\\ \nonumber & \quad + \frac  T{N-1}  \left(\sup_{p\neq 0} \frac {e_p}{|p|^2 (1-\alpha_p^2)}\right) \\  & \quad + \lambda^{-1} \frac{(N^>+1)^2}{N-1} C_1 e^{C_2}\left(\sum_q e_q  \frac{(\cosh(2\beta_q)-1)^2}{16\beta_q^2} \right) \label{u2}
\end{align}
for any $\lambda>0$.

Altogether, this shows that 
$$
H_N \leq \frac N2 \hv(0) + E^{\rm Bog} + \frac 1{1-\lambda}  U^\dagger \left( \sum\nolimits_{p} e_p\, a^\dagger_p a_p\right) U + \widetilde F_{\epsilon,\lambda}\,,
$$
with $\widetilde F_{\epsilon,\lambda}$ given by the sum of $F_\epsilon$ in (\ref{fe}), $\tfrac 32 N^> (N-1)^{-1} |E^{\rm Bog}|$ from (\ref{u1}) and the last three terms in (\ref{u2}). To complete the upper bound, we need a bound on $U \widetilde F_{\epsilon,\lambda} U^\dagger$. For this purpose, we find it convenient to first bound $\widetilde F_{\epsilon,\lambda}$ by
$$
F_{\epsilon,\lambda}\leq C_3\left( \left(\epsilon^{-1}+\lambda^{-1}\right) \frac{(T+1)^2}{N} + \left(\epsilon+N^{-1}\right) T \right)
$$
for an appropriate constant $C_3>0$. 
What remains to be shown is that
\begin{equation}\label{t2b}
U (T+1)^2 U^\dagger \leq e^{C_4} (T+1)^2
\end{equation}
for some constant $C_4>0$. 
Given (\ref{t2b}), we obtain
\begin{align}\nonumber
U H_N U^\dagger & \leq \frac N2 \hv(0) + E^{\rm Bog} + \frac{1}{1-\lambda} \sum_{p\neq 0} e_p\, a^\dagger_p a_p \\ & \quad +  C_3e^{C_4} \left( \left(\epsilon^{-1}+\lambda^{-1}\right) \frac{(T+1)^2}{N} + \left(\epsilon+N^{-1}\right) T \right) \,. \label{finu}
\end{align}
The spectrum of the operator on the right side of this inequality has exactly the desired
form. Given an eigenvalue of $\sum_{p\neq 0}e_p\, a^\dagger_p a_p$
with value $\xi$, we choose $\epsilon = O(\sqrt{\xi/N}) = \lambda$ to
obtain $\frac N 2 \hv(0)+ E^{\rm Bog} + \xi + O(\xi^{3/2} N^{-1/2})$ for the right
side of (\ref{finu}). This gives the desired upper bound.

It remains to prove (\ref{t2b}).  This can be done in essentially the
same way as in the proof of (\ref{n2b}). In fact,
$$
[X, T] = - 2 \sum_q \beta_q |q|^2 \left(b^\dagger_q b^\dagger_{-q} +b_q b_{-q}\right)
$$
and hence, similarly to (\ref{xt}),
$$
[X,T]^2 \leq  16 \left(\sum_{q}|q|^4 \beta_q^2\right) \frac{N}{N-1} \left(N^>+1\right)^2\,.
$$
Recall that $\beta_q^2 \sim \hv(q)^2/|q|^4$ for large $q$, which implies the finiteness of the sum since $v$ is bounded by assumption (and hence, in particular, square integrable on $\T^d$).  
By Schwarz,
\begin{align*}
[X,(T+1)^2]  \leq \eta (T+1)^2 + \eta^{-1} [X,T]^2
\end{align*}
for any $\eta>0$. In particular,
$$
[X,(T+1)^2] \leq C_4 (T+1)^2
$$
for some $C_4>0$. We conclude that 
$$
e^{tX} (T+1)^2 e^{-tX} \leq (T+1)^2  + C_4 \int_0^t e^{sX} (T+1)^2 e^{-sX} ds
$$
for any $t>0$. 
Iterating this bound yields (\ref{t2b}). This completes the proof of the upper bound, and hence the proof of Theorem~\ref{thm}.

\bigskip {\it Acknowledgments.} It is a pleasure to thank
J. Fr\"ohlich for inspiring discussions and for drawing my attention
towards studying the mean-field limit. Financial support by the
U.S. National Science Foundation grant PHY-0845292 is gratefully
acknowledged.


\begin{thebibliography}{19}

\bibitem{bogol} N.N. Bogoliubov, {\it On the theory of superfluidity},
J. Phys. (U.S.S.R.) {\bf 11}, 23--32 (1947).

\bibitem{calogero} F. Calogero, {\it Ground State of a One-Dimensional $N$-Body System}, J. Math. Phys. {\bf 10}, 2197--2200 (1969). {\it Solution of the One-Dimensional $N$-Body Problems with Quadratic and/or Inversely Quadratic Pair Potentials}, J. Math. Phys. {\bf 12}, 419--436 (1971).

\bibitem{CDZ} H.D. Cornean, J. Derezi{\'n}ski, P. Zi{\'n}, {\it On the infimum of the energy-momentum spectrum of a homogeneous Bose gas}, J. Math. Phys. {\bf 50}, 062103 (2009).

\bibitem{ESY} L. Erd\H os, B. Schlein, H.-T. Yau, {\it Ground-state energy of a low-density Bose gas: A second-order upper bound}, Phys. Rev. A {\bf 78}, 053627 (2008).

\bibitem{FL} J. Fr\"ohlich, E. Lenzmann, {\it Mean-Field Limit of Quantum Bose Gases and Nonlinear Hartree equation}, S\'eminaire \'E. D. P. XVIII, 26 p. (2003--2004).

\bibitem{FKS} J. Fr\"ohlich, A. Knowles, S. Schwarz, {\it On the Mean-Field Limit of Bosons with Coulomb Two-Body Interaction}, Commun. Math. Phys. {\bf 288}, 1023--1059 (2009). 

\bibitem{gir} M. Girardeau, {\it Relationship between Systems of Impenetrable Bosons and Fermions in One Dimension}, J. Math. Phys. {\bf 1}, 516--523 (1960).

\bibitem{GSlhy} A. Giuliani, R. Seiringer, {\it The Ground State Energy of the 
Weakly Interacting Bose Gas at High Density}, J. Stat. Phys. {\bf 135}, 915--934
(2009).

\bibitem{LL} E.H. Lieb, W. Liniger, {\it  Exact Analysis of an
Interacting Bose Gas. I.
The General Solution and the Ground State}, Phys. Rev. {\bf 130},
1605--1616
(1963).

\bibitem{L} E.H. Lieb, {\it Exact Analysis of an Interacting Bose  
Gas. II. The
Excitation Spectrum},
Phys. Rev. {\bf 130}, 1616--1624 (1963).

\bibitem{LSSY} E.H. Lieb, R. Seiringer, J.P. Solovej, J. Yngvason, {\it
The Mathematics of the Bose Gas and its Condensation}, Oberwolfach
Seminars, Vol. 34, Birkh\"auser (2005). Also available at arXiv:cond-mat/0610117.

\bibitem{LSol} E.H. Lieb, J.P. Solovej, {\it Ground State
Energy
of the One-Component Charged Bose Gas}, Commun. Math. Phys. {\bf 217}, 127--163 (2001). 
Errata {\bf 225}, 219--221 (2002).

\bibitem{LSol2} E.H. Lieb, J.P. Solovej, {\it Ground State
Energy of the Two-Component Charged Bose Gas}, Commun. Math. Phys.
{\bf 252}, 485--534 (2004).

\bibitem{Sol} J.P. Solovej, {\it Upper Bounds to the Ground
    State Energies of the One- and Two-Component Charged Bose Gases},
  Commun. Math. Phys. {\bf 266}, 797--818 (2006).

\bibitem{sutherland} B. Sutherland, {\it Quantum Many-Body Problem in One Dimension: Ground State}, J. Math. Phys. {\bf 12}, 246--250 (1971). {\it Quantum Many-Body Problem in One Dimension: Thermodynamics}, J. Math. Phys. {\bf 12}, 251--256 (1971). 

\bibitem{yauyin} H.-T. Yau, J. Yin, {\it The Second Order Upper Bound for the 
Ground Energy of a Bose Gas}, J. Stat. Phys. {\bf 136}, 453--503 (2009).

\end{thebibliography}
\end{document}